\documentclass[aps,floats,showpacs,amstex,amssymb,bm,tightenlines,notitlepage,hidelinks,floatfix,superscriptaddress,nofootinbib]{article}

\usepackage{dutchcal}
\usepackage{tcolorbox}
\usepackage{bm}
\usepackage[utf8]{inputenc}
\usepackage[cm]{fullpage}
\usepackage{amsmath}
\usepackage{amsthm}
\usepackage{amsfonts}
\usepackage{bigints}
\usepackage{amscd}
\usepackage{bm}
\usepackage{epsfig}
\usepackage{amssymb}
\usepackage{tabularx}
\usepackage{longtable}
\usepackage{calligra}
\usepackage{enumerate}
\usepackage{hyperref}
\usepackage{tabularx}
\usepackage{multirow}
\usepackage{mathrsfs}
\usepackage{cancel}
\usepackage[normalem]{ulem}
\def\j2{\mathbf{J}^2}
\def\a{\alpha}

\def\L{\mathcal{L}}

\def\p{\partial}

\def\k{\mathcal{k}}

\usepackage{amsmath,amsthm}
\newtheorem{thm}{Theorem}
\newtheorem{coro}{Corollary}[thm]

\newtheorem{lem}{Lemma}

\usepackage{letltxmacro}
\makeatletter
\let\oldr@@t\r@@t
\def\r@@t#1#2{%
\setbox0=\hbox{$\oldr@@t#1{#2\,}$}\dimen0=\ht0
\advance\dimen0-0.2\ht0
\setbox2=\hbox{\vrule height\ht0 depth -\dimen0}%
{\box0\lower0.4pt\box2}}
\LetLtxMacro{\oldsqrt}{\sqrt}
\renewcommand*{\sqrt}[2][\ ]{\oldsqrt[#1]{#2}}
\makeatother

\begin{document}
\title{Black hole regions containing  no trapped surfaces}
\author{Gustavo Dotti \\
Facultad de Matemática, Astronomía, F\'{\i}sica y Computaci\'on (FaMAF), Universidad Nacional de C\' ordoba  \\
and Instituto de F\'{\i}sica  Enrique Gaviola, CONICET. Ciudad Universitaria, (5000) C\'ordoba, Argentina.\\
gdotti@famaf.unc.edu.ar}

\maketitle

\begin{abstract}
A simple criterion is given to rule out the existence of closed trapped surfaces in large open regions 
inside black holes. 
\end{abstract}

\tableofcontents

\section{The importance of closed trapped surfaces}
\label{surfaces}

The drawbacks of the standard definition \cite{onil,waldbook} of a  black hole spacetime $(M,g_{ab})$ and its black hole  region $B$ 
have been extensively discussed. 
The major problem is its global character: 
knowledge of  the \textit{entire} spacetime $(M,g_{ab})$ is required 
to determine if a point belongs to the black hole region  $B \subset M$, 
since $B$ is defined as a set \textit{causally disconnected from far away regions} (technically, $B = M-J^-(\mathscr{I}^+)$, where 
$J^-(\mathscr{I}^+)$ is the causal past of future null infinity $\mathscr{I}^+$ \cite{onil,wald}), 
something that cannot be  tested on a \textit{partial} (that is, extendable) solution   $(M',g_{ab})$ of Einstein's equations.
Numerical codes used to simulate strong gravity processes work  by integrating Einstein's equations  
a spacelike slice at a time. Questions such as if a black hole is \textit{being formed}  lack sense under the above definition of $B$.
In practice, what is done is to search for  \textit{closed trapped 
surfaces} in every newly generated time slice $\Sigma$  \cite{Andersson}. Since these surfaces can only exist inside $B$ 
(proposition 12.2.2 in \cite{waldbook}),  the boundary of the subset  of $\Sigma$ 
containing closed trapped surfaces (which is, under certain conditions, a \textit{marginally outer trapped surface} \cite{Andersson})
 is a proxy  for  the intersection with $\Sigma$  of the event horizon $H_e= \p B$. 
The slicing dependence, stability and dynamics of these quasi-locally defined horizons,  
which  lie within $H_e$, is
a subtle issue discussed, e.g., in \cite{Andersson} and more recently in \cite{Booth:2021sow,Pook-Kolb:2021jpd} (see also 
\cite{Senovilla:2023npd} for the the case of a positive cosmological constant).  

Although  closed trapped surfaces are \textit{the} black hole signature and, being a quasi-local concept, 
offer a sensible approach to the issue of searching  the black hole boundary 
$\p B$, 
it has been known for a long time that there are large open regions in $B$  admitting no such surfaces.
 This fact was first pointed out in 
\cite{wald}, where a Cauchy slicing  of the Kruskal manifold  was constructed which, in spite of getting 
 arbitrarily close to the $r=0$ singularity,  
eludes all closed trapped surfaces. Not only there are no closed trapped surfaces  contained in any of its slices, but also 
the causal past of any of these Cauchy hypersurfaces $\Sigma$ contains no such a  surface, no matter how close is 
$\Sigma$ from the singularity. 
A numerical relativist unfortunate enough to pick such a slicing would proceed unaware of the fact that there is a black hole region. 
The proof  in \cite{wald} is based on the fact that there are no closed trapped surfaces 
in the intersection $J^-(\gamma) \cap X_{Sch}$, where $\gamma$ is  a particular timelike curve 
reaching the $r=0$ singularity and $X_{Sch}$ the 
black hole open subset defined using standard Schwarzschild coordinates $(t,r,\theta,\phi)$ as 
\begin{equation}\label{xset}
X_{Sch} = \left\{ (t,r,\theta,\phi) \; | \; r<2M, \cot(\theta) >\frac{M-r}{ \sqrt{r (2M-r)}} \right\}. 
\end{equation}
Inspection of the proof 
in \cite{wald} reveals that, in fact, there are no closed trapped surfaces contained in the \textit{entire} set $X_{Sch}$.  
A  closed trapped surface $S$ might \textit{enter} $X_{Sch}$, that is, it is possible that $S \cap X_{Sch}$ be non empty, 
but it is impossible that $S \subset X_{Schw}$. 

In this paper we prove some  results (Theorem \ref{main} and its corollaries in Section \ref{ctrots}) that 
 allow to anticipate obstructions such as the impossibility of finding closed trapped surfaces in (\ref{xset}). 
 We  give a tool that allows to find 
sets like (\ref{xset})  in an arbitrary spacetime $M$ (see corollaries \ref{coro1} and  \ref{coro2} in section \ref{ctrots}):\\

\textit{ Assume that  there is a function $g: M \to \mathbb{R}$ such that $\nabla^a g$ is future null on a domain $D \subset M$. 
Define $X_g \subset D$ as the subset of $D$ where $\Box  g >0$. If $S$ is a closed trapped or marginally trapped surface,  
 it is not possible that $S \subset X_g$. Moreover, if $S \cap X_g \neq \emptyset$, $g|_S$ cannot attain a local 
maximum within $X_g$.} \\

To make this paper self-contained we introduce all basic concepts in the following section, stressing 
the relation that exists between \textit{surfaces} (here defined as codimension two, spacelike submanifolds) and null hypersurfaces, 
in particular, the fact that a null hypersurface is a  bundle of  null geodesics orthogonal to a surface. 
In section \ref{ctrots}  the 
  main results are proven. An appendix serves as a quick reference for the notation used through the text. 
Among the examples and applications given in section \ref{apps} we find large open subsets in the region between the horizons of 
a  Kerr black hole where there are no 
closed trapped surfaces.

\section{Null hypersurfaces and their spacelike sections} \label{concepts}

Our discussion does not get more involved in arbitrary dimensions, so we proceed by defining 
 a \textit{spacetime} $(M,g_{ab})$ as  an oriented $n+1$ dimensional Lorentzian manifold ($n \geq 2$, mostly plus signature convention) 
 which is 
time oriented. 
A hypersurface is a smooth embedded submanifold of $M$ of dimension $n$. Locally, it can be given as a level set  
of a smooth function $\phi:M \to \mathbb{R}$. 
The hypersurface $N$ is null if there is a future null vector field $l^a$ on $N$ such that, for every $p \in N$,  
$T_pN = \{ X^a \in T_pM \, | \; X^a l_a|_p =0 \}$. In particular, since $l^a l_a =0$, 
$l^a$ is both orthogonal and tangent to $N$. The restriction of the metric 
to $T_pN \otimes T_pN$ is a degenerate symmetric tensor of signature $(0++...+)$, that is,  
the space of null vectors in $T_p N$ is one dimensional, 
 the unique null direction being that of $l^a$.  If $X^a \in T_pN$, then $X^aX_a \geq 0$, with equality 
holding only 
if $X^a \propto l^a |_p$. 
If properly scaled, $l^a$  will satisfy the  affine geodesic equation. To prove this, let  $O \subset M$ be an open set such that 
$N \cap O$ 
is a level set of a function $\phi: O \to \mathbb{R}$. In $N \cap O$,  $l^a \propto \nabla^a \phi$, in particular, $\nabla^a \phi \nabla_a 
\phi |_{N}=0$. Given 
    any vector $X^a \in T_p N$, 
$X^a (\nabla^b \phi)\nabla_b (\nabla_a \phi) = \tfrac{1}{2} X^a\nabla_a (\nabla^b \phi \nabla_b \phi)=0$, which means that 
$ (\nabla^b \phi)\nabla_b (\nabla^a \phi)$ is normal to $T_p N$, then parallel to $\nabla^a \phi$, that is 
$ (\nabla^b \phi)\nabla_b (\nabla^a \phi) \propto \nabla^a \phi$ at points of $N$. 
From now on, 
we will assume that the future null vector field 
$l^a$  satisfies the affine geodesic equation. $N$ can then be regarded as a bundle of null geodesics: the integral lines of $l^a$. 
These null geodesics are called \textit{generators } of $N$.  The affine parametrization is non unique, if $l^a$ is affine then so is 
$\tilde l^a = f l^a$, where $f:N \to \mathbb{R}^+$ is constant on generators. 

\begin{lem} \label{gaf} If a  function $g:= O \to \mathbb{R}$ is such that 
 \emph{all} its level sets  are null hypersurfaces,  then  the null vector field $\nabla^a g$ satisfies the \textit{affine} geodesic equation 
 (no point-wise re-scaling required).
 \end{lem}
 \begin{proof}
 \begin{equation}\label{age}
(\nabla^c g) \nabla_c \nabla^a g = (\nabla^c g) \nabla^a (\nabla_c g) = \frac{1}{2} \nabla^a(\nabla^c g \nabla_c g)=0.
\end{equation}
\end{proof}

Given \textit{any} set of $n-1$ linearly independent  \textit{spacelike} vectors $\{ e_1^a,..., e_{n-1}^a \}$ at $p \in N$
and a choice of affine $l^a$, we  define the divergence $\Theta_{N_l}$ of $N$ at $p$ as 
(we use letters from the beginning of the alphabet as tensor indices, from the middle of the alphabet to number basis elements)
\begin{equation}\label{tita}
\Theta_{N_l}(p)=  g_{ac} \, h^{ij} \,  e^c_i\,  e^b_j \, \nabla_b l^a,
\end{equation}
 where $h^{ij}$ is the inverse of the matrix $h_{kl}= g_{ab} e_k^a e_l^b$ (that is, $h^{ij}h_{jk}=\delta^i_k$). 
 Note that,  although  the vector field  $l^a$ is defined  only on $N$, $\Theta_{N_l}(p)$ is well defined because 
 all directional derivatives in (\ref{tita}) are along directions tangent to $N$. Less obvious is the fact that 
 the right hand side of (\ref{tita}) is independent of the chosen $n-1$ linearly independent spacelike vectors $e^a_j$, 
 even if they span  different spacelike subspaces of $T_pN$. To prove this note that any alternative set 
 $\tilde e_i^a \in T_pN$, $i=1,2,...,n-1$ of linearly independent spacelike vectors satisfies $\tilde e_i^a = R^k_i e_k^a + \a_i l^a$ 
 with $R^k_i \in GL(n-1,\mathbb{R})$, then $\tilde h_{ij} =  g_{ab} \tilde e_i^a \tilde e_j^b = g_{ab} (R^k_i e^a_k+ \a_i l^a) (R^l_j e^b_l + \a_j l^b) 
=g_{ab} e^a_k e^b_l R^k_i R^l_j = h_{kl} R^k_i R^l_j$. In matrix notation (left or upper indices label file), this reads $\tilde{\bm{h}}
= \bm{R}^T \bm{h} \bm{R}$, and implies $\tilde{\bm{h}}^{-1}=  \bm{R}^{-1} \bm{h}^{-1}{\bm{R}^T}^{-1}$. As a consequence
 \begin{equation}
\begin{split}
\tilde h^{ij} \tilde e_i^c \tilde e_j^b \nabla_b l_c &= (\bm{R}^{-1} \bm{h}^{-1}{\bm{R}^T}^{-1})^{ij} (R^k_i e_k^c+ \a_i l^c) 
 (R^l_j e_l^b+ \a_j l^b) \nabla_b l_c \\
 &=  (\bm{R}^{-1} \bm{h}^{-1}{\bm{R}^T}^{-1})^{ij} R^k_i e_k^c  R^l_j e_l^b \nabla_b l_c\\
 &= h^{kl} e_k^c e_l^b \nabla_b l_c.
 \end{split}
 \end{equation}
  Conventionally \cite{onil}, one picks an orthonormal basis, 
 then  (\ref{tita}) simplifies to $\sum_{i=1}^{n-1} g_{ac}\,  e^c_i\,  e^b_i \, \nabla_b l^a$.  
Note that  if we re-scale 
\begin{equation}\label{rs1}
l^a \to \tilde l^a = fl^a,
\end{equation}
$f: N \to \mathbb{R}^+$ constant on generators (so that $\tilde l^a$ is 
 also affine geodesic and future pointing), then 
 \begin{equation} \label{rs2}
 \Theta_{N_l} \to  \Theta_{N_{\tilde l}} = f \Theta_{N_l}.
 \end{equation}
  In particular, we may say without ambiguities that 
 $N$ (its generators) diverges (diverge) towards the future if $\Theta_{N_l} >0$. We will similarly say that $N$ converges towards the future if 
 $\Theta_{n_l}<0$. As an example, the black hole event horizon $H_e =\p B$ of a spacetime $M$ satisfying certain energy conditions  is future non-convergent 
  (Hawking area theorem). \\

 \begin{lem} \label{divL} Let $\L^a$ be \emph{any} null vector field that extends the chosen affine tangent $l^a$ of the generators of a null hypersurface $N$ 
 to an open neighborhood $O \subset M$ of   $p \in N$, then, 
 in $O \cap N$,
 \begin{equation} \label{div}
 \Theta_{N_l} = \nabla_a \L^a 
 \end{equation}
 \end{lem}
 \begin{proof}
 At $p \in N$ complete the set $\{e_1^a,...,e_{n-1}^a\}$ used in (\ref{tita}) to a basis $\{e_1^a,...,e_{n-1}^a, \ell^a, \mathcal{k}^a \}$, where $\ell^a=l^a|_p$ 
 and $\mathcal{k}^a$ is the only future oriented null vector  orthogonal to the $e_j^a$ and satisfying $\mathcal{k}^a \ell_a=-1$. 
 Note that $\mathcal{k}^a$ points out of $N$ and that $g^{ab}|_p= h^{ij} e_i^a e_j^b - \mathcal{k}^a \ell^b - \ell^a \mathcal{k}^b$, then 
 \begin{equation}\label{div0}
 \begin{split}
\nabla^a \L_a  &= g^{ab} \nabla_a \L_b \\&= h^{ij} e_i^a e_j^b \nabla_a \L_b - \mathcal{k}^a \ell^b \nabla_a \L_b -
\ell^a \mathcal{k}^b \nabla_a \L_b \\
&=  g_{bc} h^{ij} e_i^a e_j^b \nabla_a \L^c - \mathcal{k}^a \L^b \nabla_a \L_b - \mathcal{k}^b l^a \nabla_a l_b \\
&=  g_{bc} h^{ij} e_i^a e_j^b \nabla_a l^c - \tfrac{1}{2} \mathcal{k}^a \nabla_a (\L^b \L_b)  \\ &= \Theta_{N_l}(p),
\end{split}
\end{equation}
where we used that $l^a$ is affine geodesic and its extension $\L^a$ away from $N$ is a null vector field.
\end{proof}
The fact that  $\Theta_{N_l}$ can be calculated by taking the divergence of \textit{any} null vector field  $\L^a$ that extends $l^a$ away from $N$ 
is (sometimes) convenient to perform explicit calculations. In the particular case where $N$ is a level set of a function $g$ all whose level sets are null, 
$\nabla^a g$ is null and we may assume (flipping sign if necessary), that is future pointing. 
If we choose $l^a=\nabla^a g |_N$ (equation (\ref{age})),  an obvious null extension of $l^a$ (which is  also geodesic everywhere) 
is $\L^a = \nabla^a g$. Applying the previous lemma to this case gives:

\begin{lem} \label{divL2} 
If $N$ is a level set of a function $g$ whose gradient $\nabla^a g $ is future null everywhere, and we choose $l^a = \nabla^a g|_N$, 
then
\begin{equation}\label{Box}
\Theta_{N_l} = \nabla_a \nabla^a g |_N =: \Box g |_N 
\end{equation}
\end{lem}

Of course, we might have normalized  the tangent to the $N$ generators as $\tilde l^a= f \nabla^a g |_N$ 
with $f$ a positive function constant on generators. In that case, a combination of the re-scaling equations (\ref{rs1}) and (\ref{rs2}) with 
equation (\ref{Box})  gives
\begin{equation} \label{gg}
\Theta_{N_{\tilde l}} = f \, \Box g |_N
\end{equation}

By a  \textit{surface} in $M$ we mean  a codimension two spacelike submanifold $S$. At every point of $S$ 
the orthogonal to  the tangent space  has induced metric of signature $(-,+)$, so (at least locally) we can define 
 two future  null vector fields  $\ell^a$ and $\mathcal{k}^a$ on the normal bundle of $S$ satisfying  $\mathcal{k}^a \ell_a=-1$. 
These are unique up to flipping and  re-scaling $\ell^a \to \lambda \ell^a$, $\mathcal{k}^a \to \lambda^{-1} \mathcal{k}^a$, $\lambda:  
S \to \mathbb{R}^+$. 
In some contexts it makes sense to call one of these future null vector fields \textit{outgoing} and the other  \textit{ingoing}. 
If we integrate the geodesic equation with initial condition $\ell^a$ from every point of $S$ and take the union of these null geodesics, 
we get (at least, near $S$) a null hypersurface $N^S_l$  with affine generators $l^a$  satisfying 
$l^a|_S=\ell^a$, of which  $S$ is a \textit{proper transverse 
 section}, that is,  
an $n-1$ dimensional spacelike submanifold. $N^S_k$ and $k^a$ are defined analogously using $\mathcal{k}^a$. Note that \textit{any} 
 null hypersurface  locally agrees with 
$N^S_l$ [$N^S_k$] for some proper transverse section $S$.  On $S$ we define 
\begin{equation} \label{tls}
\Theta_\ell^S := \Theta_{N^S_l} \mid_S
\end{equation}
and similarly $\Theta_{\k}^S := \Theta_{N^S_k}\mid_S$. 
 The \textit{mean curvature vector field} on $S$ (here defined  following  the overall sign and normalization  conventions  in \cite{barriers}) is 
\begin{equation}
H^b = -h^{ij}  \left( e_i^a \nabla_a e_j^b \right)^\perp,
\end{equation}
where $\perp$ means  the component normal to $S$. The definition is independent of the local basis of vector fields $\{ e_i^a, i=1,2,...,n-1 \}$ on $TS$. 
Note that 
\begin{equation}
H_b \ell^b = -h^{ij}  \left( e_i^a \nabla_a e_j^b \right) \ell_b = h^{ij} e_i^a e_j^b (\nabla_a \ell_b) = \Theta_{\ell}^S
\end{equation}
and similarly $H_b \k^b =\Theta_{\k}^S$, then 
\begin{equation}\label{Htita}
H^a = h^{ij}  \left( e_i^a \nabla_a e_j^b \right) (\ell^a \k_b + \k^a \ell_b) = -\Theta_{\ell}^S \k^a -\Theta_{\k}^S \ell^a,
\end{equation}
For a more direct  and natural definition of the mean curvature vector field of arbitrary codimension 
semi-Riemannian submanifolds of 
a semi-Riemannian manifold see \cite{barriers,onil}.\\

\noindent
\underline{Example}: 
\textit{Take $M= \mathbb{R}^{2+1}$, three dimensional Minkowski spacetime, with metric 
$ds^2=-dt^2+dx^2+dy^2$.
Let $a>b>0$ and consider the ellipse $S$ defined by 
\begin{equation}\label{ellipse}
(t,x,y)=(0,a \cos \theta,b \sin \theta), \;\;\theta \in [0,2\pi].
\end{equation}
At every point of $S$ we can determine the two future null directions, there 
is a sensible notion of ``outgoing'' and ``ingoing'', for which we may take respectively 
the vector fields (normalized as above) 
\begin{equation}\label{lnS}
\begin{split}
\ell (\theta) &=    \p_t +  K(\theta)^{-1}[b \cos \theta \; \p_x + a \sin \theta \; \p_y ]\\
\k (\theta) &=  \tfrac{1}{2}  \p_t - \tfrac{1}{2}  K(\theta)^{-1} [b \cos \theta \; \p_x - a \sin \theta \; \p_y ]
\end{split}
\end{equation}
where $K(\theta)= \sqrt{a^2\sin^2 \theta + b^2 \cos^2 \theta}$. \\
 The null surface $N_l^S$ is defined as the bundle of outgoing null geodesics with initial condition $\ell(\theta)$. 
It can be  parametrized 
using $\theta$ and an affine parameter $s$ along the geodesics:
\begin{equation}\label{exNl}
(\theta,s) \to  (t=s,x= (a+sb/K(\theta)) \cos \theta,y= (b+sa/K(\theta)) \sin \theta)
\end{equation}
The expression for the geodesic field $l^a(s,\theta)$ on $N_l^S$ is independent of $s$ and 
agrees with the right hand side of the first equation (\ref{lnS}), 
and similarly for $k^a(s,\theta)$ on $N_k^S$. On $S$, a unit tangent vector field is 
\begin{equation}
e_1^a(\theta)= -(a/K(\theta)) \sin \theta \; \p_x + (b/K(\theta)) \cos \theta \; \p_y
\end{equation}
Using the definition (\ref{tls}) and (\ref{tita}) we can calculate $\Theta_\ell^S= e_{1a} e_1^b\nabla_b \ell^a$ 
as a covariant derivative in $M=\mathbb{R}^{2+1}$. To this end we may use   \emph{any} extension 
of $\ell^a(\theta)$ and $e_1^a(\theta)$, around $S$ (e.g.,   $K(x,y)^{-1} [(b/a) x \, \p_x + (a/b) y \, \p_y] +\, \p_t$ and 
$K(x,y)^{-1} [(b/a) x \, \p_y -(a/b)y \, \p_x]$, where $K(x,y)= \sqrt{ (b/a)^2 x^2 + (a/b)^2 y^2}$), and then evaluate on $S$. 
The result is 
\begin{equation}\label{tital}
\Theta_{\ell}^S= \frac{ab}{(a^2 \sin^2 \theta + b^2 \cos^2 \theta )^{3/2}}.
\end{equation}
We similarly obtain $\Theta_\k^S= -\tfrac{1}{2} \Theta_\ell^S$, then 
\begin{equation}\label{HS} 
H = \frac{a b^2 \cos \theta}{a^2 \sin^2 \theta+ b^2 \cos ^2 \theta} \p_x + \frac{a^2 b \sin \theta}{a^2 \sin^2 \theta+ b^2 \cos ^2 \theta} \p_y
\end{equation}
\hspace*{9cm} \hrulefill \vspace{1cm}
}

A surface $S$ satisfies the  \textit{trapping condition (marginal trapping condition)} at $p \in S$ if the divergences 
 $\Theta^S_\ell(p)$ and $\Theta^S_\k(p)$ 
are both negative (non-positive).
Other related concepts turn out to be useful, particularly that of 
marginally \textit{outer} trapped ($\Theta_\ell^S(p)=0$ and no condition  on $\Theta_\k^S(p)$, in contexts where 
there is a notion of $\ell^a$ being outer pointing, see \cite{Andersson}). 
If such condition is satisfied at every point of $S$ we say that the surface is trapped, marginally trapped, etc.  
Note from (\ref{Htita}) that  these conditions can be re-stated as $H^a$ being timelike, causal, along $\ell^a$, etc.

By a \textit{closed manifold} we mean, as usual,  an ordinary manifold (that is, without boundary)  which is compact. 
 The relevance of the mean curvature vector 
field on a closed surface $S$ comes from the following fact \cite{onil}: 
if $\zeta^a$ is any vector field on $M$, $S_t$ the image of $S$ under the flow $\Phi_t: M \to M$ 
of this vector field and $A(S_t)$ the area of $S_t$ (that is, its $n-1$ volume, which is finite since $S$, and then $S_t$ for small enough $t$, 
are compact), 
then 
\begin{equation}\label{evo}
\left. \frac{d A(S_t)}{dt} \right|_{t=0} = \int_S H^a \zeta_a dS.
\end{equation}
Closed trapped surfaces are 
codimension two spacelike closed manifolds that satisfy the trapping condition at every point. 
From (\ref{evo}) follows that for $S$  closed trapped, since $H^a$ is a timelike future vector at every point of $S$, 
the area shrinks under the flow of \textit{any} future causal vector field $\zeta^a$. If a spacetime 
$M$ contains a black hole $B$ and $S$ is a  closed trapped surface, then  $S \subset B$ 
(proposition 12.2.2 in \cite{waldbook}). 
The non trivial character of closed trapped surfaces is best exemplified by the existence of 
 closed trapped surfaces entering flat regions of $B$, as constructed  explicitly, e.g., 
in \cite{Bengtsson}.

\section{Criteria to rule out closed trapped surfaces}\label{ctrots}

We address now the problem of finding  sets such as (\ref{xset}),  within which 
closed trapped surfaces are not allowed. 

\begin{thm} \label{main} Assume there is a $C^2$ function $g: M \to \mathbb{R}$ such that $\nabla^a g$ is future null 
on a domain $D \subset M$, $S \subset M$ is a spacelike surface intersecting $D$ and $p \in S$ is a critical point of the restriction 
$g|_S$.
\begin{enumerate}[i)]
\item  $\nabla^a g |_p$ is orthogonal to $S$.
\item  If we choose the future null field $\ell^a$ on $S$ such that $\ell^a |_p=\nabla^a g|_p$ then 
\begin{equation}\label{se}
\Theta^S_{\ell} |_p = \Box g |_p - \Delta_S g |_p , 
\end{equation}
where $\Delta_S$ is the Laplacian of $S$.
\end{enumerate}
\end{thm}
\begin{proof}
Since $p$ is a critical point of $g|_S$, for every $X^a \in T_pS$,  $X^a \nabla_a g=0$. This implies that 
$\nabla^a g|_p$ is orthogonal to $T_pS$ and  along one of the two orthogonal future null directions of $S$ at $p$. Choose 
future null vector fields $\ell^a$ and $\mathcal{k}^a$ on $S$ normalized as usual ($\mathcal{k}^a \ell_a=-1$) 
and such that $\ell^a |_p = \nabla^a g|_p$. 
 To prove (\ref{se}), let  $O \subset M$ be an 
open neighborhood of $p$ admitting  a  basis of vector fields $\{ E_i^a, L^a, K^a  \, | \, i=1,...,n-2 \}$ 
with the $E_i^a$ orthonormal and $E_i^a|_S=:e_i^a$ tangent to $S$, $L^a$ and $K^a$ null, perpendicular  to the $E_i^a$ and 
satisfying $L_a K^a=-1$ and $L^a|_S=\ell^a$ (thus, $K^a|_S=\mathcal{k}^a$). 
Note that, in $O$, the inverse metric can be written as
\begin{equation}
g^{ab}= E_i^a E_i^b -L^a K^b - K^a L^b
\end{equation}
and then 
\begin{equation}\label{qq}
\nabla^a g = ( E_i^a E_i^b -L^a K^b - K^a L^b) \nabla_b g =: Z^a - g_K L^a -g_L K^a,
\end{equation}
from where 
\begin{equation}\label{Lspan}
L^a = -\frac{1}{g_K} \left[ \nabla^a g -Z^a + g_L K^a \right].
\end{equation}
Apply to the above equality the linear differential operator $E_{ia} E_i^b \nabla_b$:
\begin{equation}\label{pe}
\begin{split}
E_{ia} E_i^b \nabla_b L^a =& - E_{ia} E_i^b \nabla_b ({g_K}^{-1})  \left[ \nabla^a g -Z^a + g_L K^a \right] \\ 
 & -\frac{1}{g_K} \left[ E_{ia} E_i^b \nabla_b \nabla^a g - E_{ia} E_i^b \nabla_bZ^a + K^a E_{ia} E_i^b \nabla_bg_L +
 g_L E_{ia} E_i^b \nabla_b K^a\right]\\
 =& -\frac{1}{g_K} \left[ E_{ia} E_i^b \nabla_b \nabla^a g - E_{ia} E_i^b \nabla_bZ^a  +
 g_L E_{ia} E_i^b \nabla_b K^a\right]
\end{split}
\end{equation}
From (\ref{Lspan}) and $L^a|_p = \ell^a|_p=\nabla^a g |_p$ 
 follows that $g_k(p)=-1$ and $g_L(p)=0$. Evaluating equation (\ref{pe}) at $p$  then gives:
\begin{equation}\label{pp}
e_{ia} e_i^b \nabla_b \ell^a \mid_p = e_{ia} e_i^b \nabla_b \nabla^a g \mid_p - e_{ia} e_i^b \nabla_bZ^a \mid_p
\end{equation}
Consider now the null hypersurface $N[g(p)]:= g^{-1}(g(p))$. 
The vector field $\L^a = \nabla^a g$ on $M$ can be thought of as a null  extension of the affine tangent $\nabla^a g|_{N[g(p)]}$, 
so Lemma \ref{divL} applies and we recognize that the first term on the right side of (\ref{pp}) equals 
$\nabla_a \L^a |_p = \Box g|_p$ (and also equals $\Theta^{N[g(p)]}_{\nabla^a g} \mid_p$). 
From (\ref{qq}) follows 
that $Z^a|_S$ is the orthogonal projection of the gradient $\nabla^a g$ onto $S$, then the second term on the right in (\ref{pp}) is (minus) 
$\Delta_S g |_p$. Finally, (\ref{tls}) and (\ref{tita}) show that the left side of (\ref{pp}) is $\Theta^S_\ell$. 
It then follows that (\ref{pp})  translates  into (\ref{se}).
\end{proof}

\begin{coro} \label{coro1} Assume $g: M \to \mathbb{R}$ is $C^2$ and such that $\nabla^a g$ is future null on a domain $D \subset M$. 
Define   $X_g$ as the subset of $D$ where $\Box g >0$. 
If $S$ is a surface intersecting $X_g$ 
and $g|_S$ 
has a local maximum at $p \in X_g$ then  $S$ cannot satisfy the trapping or marginally trapping condition  at $p$.
\end{coro}

\begin{proof} In view of the local maximum condition $\Delta_S g |_p \leq 0$, then 
\begin{equation}\label{se2}
\Theta^S_{\ell} |_p = \Box g |_p - \Delta_S g |_p  \geq \Box g|_p >0.
\end{equation}
\end{proof}

\begin{coro} \label{coro2} 
Assume $g: M \to \mathbb{R}$ is $C^2$ and such that $\nabla^a g$ is future null on a domain $D \subset M$. 
Define   $X_g$ as the subset of $D$ where $\Box g >0$.  
If $S$ is a closed trapped or marginally trapped surface it is not possible that $S \subset X_g$.
\end{coro}

\begin{proof} The compactness of $S$ implies $g|_S$ has a global, then a local maximum. 
\end{proof}

\noindent
\underline{Example (continued)}: \textit{Take $g=\sqrt{x^2+y^2}-t$, which has $\nabla^a g = \p_t + \frac{x}{r} \p_x + \frac{y}{r} \p_y$,
 $r=\sqrt{x^2+y^2}$. This 
is future null in the domain $r>0$ where $g$ is $C^2$. The restriction of $g|_S= \sqrt{a^2 \cos^2 \theta+b^2 \sin^2 \theta}$ has local maxima 
at $\theta=0, \pi$ and minima at $\theta=\pm \pi/2$. Let us analyze the local maximum at $\theta=0$, that is, the point $p$ 
with coordinates $(t=0,x=a,y=0)$. The induced metric on $S$ is $ds^2 = z^2(\theta)\, d\theta^2$, 
where $z(\theta)= \sqrt{a^2 \sin^2 \theta + b^2 \cos^2 \theta}$, then  
\begin{equation}
\Delta_S g |_S = \frac{1}{z(\theta)} \p_\theta \left[ \frac{1}{z(\theta)} g|_S \right] = -\frac{a^2b^2 (a^2-b^2) \cos(2 \theta)}{(a^2 \cos^2 \theta 
+ b^2 \sin^2 \theta)^{3/2} (a^2 \sin^2 \theta + b^2 \cos^2 \theta)^2}
\end{equation}
Now $\Box g = (x^2+y^2)^{-1/2}$, then 
\begin{equation}
\Box g |_p - \Delta_S g |_p = \frac{1}{a} -  \frac{b^2-a^2}{a b^2}
\end{equation}
which agrees with (\ref{tital}) evaluated at $\theta=0$. Note that $S$ cannot be trapped since  $S\subset X_g$ (we already knew $S$ was not trapped, 
since $H^a$ is spacelike for this surface, see (\ref{HS})).
\hspace*{9cm} \hrulefill \vspace{1cm}
}

\subsection{Geometric interpretation}

Corollary \ref{coro2} translates into more geometrical terms as follows:

\begin{coro} A closed trapped surface S cannot lie entirely within an open set $X$  that is foliated by future diverging null hypersurfaces.
\end{coro}
\begin{proof}
The leaves of the foliation are level sets of a function $g:X \to \mathbb{R}$ with $\nabla^a g$ null. Changing $g \to -g$ if necessary, we may assume that 
$\nabla^a g$ is future and take $l^a = \nabla^a g\mid_N $ as the affine generator of the level set $N$. 
Since all level sets are future diverging, in view of (\ref{Box}) $\Box g>0$ on $X$, then Corollary \ref{coro2} applies to $X=X_g$.
 \end{proof}

In some simple cases this result allows  
 a rapid identification of highly symmetric $X$ sets by inspection. As an example, consider 
the  Reissnser-Nordström metric $ds^2=-f(r) dv^2+ 2 dv\, dr + r^2 (d \theta^2+ \sin^2 \theta \, d\phi^2)$, where $f(r)=(r-r_e)(r-r_i)/r^2$ and 
$X$ defined by the condition $0<r_i<r_e$. A conformal diagram is shown in Figure  \ref{zzz}, where some level sets of $r$ are displayed as thin black lines. 
The $r$ level sets (which are, generically, non-null) are transverse to the null hypersurface foliation we are interested in,  
and allow to determine by inspection if these null hypersurfaces have 
positive or negative divergence. 
 The inner region $0<r<r_i$ 
is foliated by spherically symmetric null hypersurfaces (thick gray lines in the figure). These are easily seen to 
 diverge towards the future, since they have 
proper cross 
sections which are spheres with radii growing from zero to $r_i$. We conclude that it is impossible to fit a trapped surface 
(any size, any geometry) within this region. On the other hand, as is well known, the region $r_i<r<r_e$ admits closed trapped surfaces: 
a calculation shows that any sphere of constant $v$ and $r$ 
  is trapped. This result can also be proved by inspection: a sphere $S$ like this is represented by a point in this diagram, and the null bundles 
$N_\ell^S$ and $N_\k^S$, which are spherically symmetric, are represented   by the gray thick lines emerging from $S$. 
Both bundles  evolve to 
smaller $r$ values towards the future, so that they  have negative expansion everywhere, then at $S$, showing that $S$ is trapped.\\

\begin{figure}\begin{center}
\includegraphics[scale=0.2]{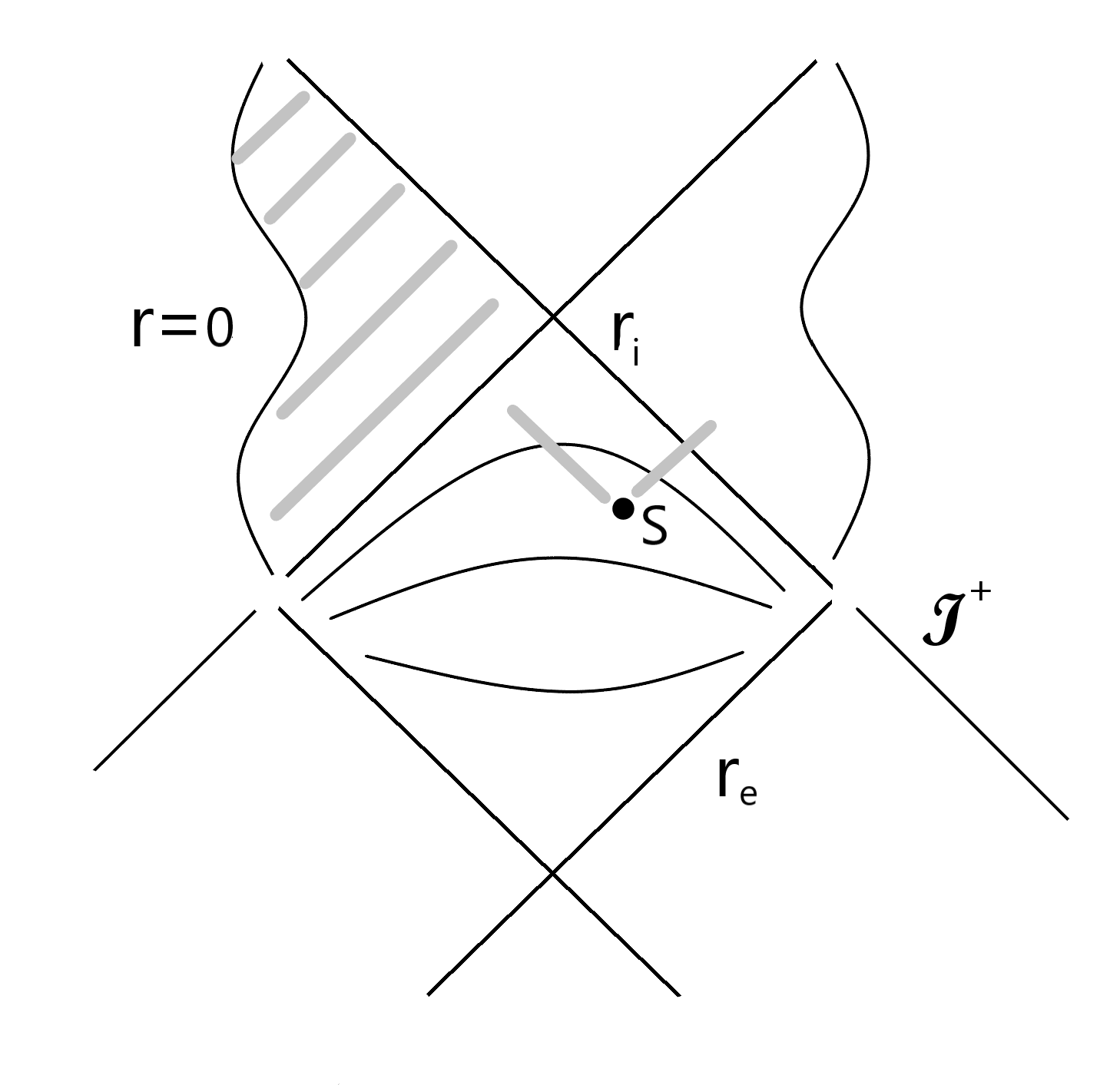}
\end{center}
\caption{The conformal diagram for the Reissnser-Nordström metric. Qualitative knowledge of the level sets of the area radius $r$ (thin black lines) 
allows to 
anticipate that there are no closed trapped surfaces in  the inner region $0<r<r_i$, as this is foliated by spherically symmetric null hypersurfaces 
diverging towards the future (thick gray lines in this region). The sphere $S$ in the region defined by $r_i<r<r_e$, on the other hand, is trapped, 
since both bundles of normal null geodesics (thick gray lines emerging from $S$) evolve towards smaller $r$ regions. \label{zzz}}
\end{figure}

The proof of Theorem \ref{main}  brings our attention to 
the null hypersurface $N[g(p)]=g^{-1}(g(p))$ (see the paragraph following equation (\ref{pp})), for which the restriction of 
$\nabla^a g$ is tangent to affine generators. In view of our previous results, $\Theta_{N[g(p)]_{\nabla^a g}}= \Box g |_{N[g(p)]}$, then 
equation (\ref{se}) can be written as 
\begin{equation} \label{comp1}
\Theta_{N^S_l}\mid_p - \Theta_{N[g(p)]_{\nabla^a g}}\mid_p = -\Delta_S g \mid_p, 
\end{equation}
 Equation (\ref{comp1}) gives a comparison of the divergence of two null 
surfaces intersecting  at a point $p$, one of them a level set of $g$. In particular, it says that 
$\Theta_{N^S_l}\mid_p \geq \Theta_{N[g(p)]_{\nabla^a g}}$ is $p$ is a local maximum of $g$. 
 Since $\nabla^a g(p)$ gives the orthogonal null direction 
for both  $T_p N^S_l$ and $T_p N[g(p)]$, the intersection of these two null hypersurfaces is tangential at $p$. 
In general, if two null hypersurfaces $N_\a$, $\a=1,2$ with future affine null geodesic fields $l_\a^a$ intersect \textit{tangentially} at a point $p$, then 
$(l_1^a(p))^\perp = T_pN_1=T_pN_2 = (l_2^a(p))^\perp$, and this  implies $l_2^a(p) \propto l_1^a(p)$. As a consequence, the geodesics with initial condition 
$l_2^a(p)$ and $l_1^a(p)$ agree and therefore  $N_1$ and $N_2$ share an open segment of this generator around $p$,  along which they intersect tangentially. 
This situation is studied in  \cite{Galloway1,Galloway2} where its is shown that, if a timelike hipersurface $P$ is chosen that intersects both null surfaces 
\textit{properly} (meaning that $\Sigma_1 = P \cap  N_1$ and $\Sigma_2 = P \cap \subset N_2$ are spacelike submanifolds), and 
 $p \in N_1 \cap N_2 \cap P = \Sigma_1 \cap \Sigma_2$, 
 there is a one to one map $h: \Sigma_1 \to \Sigma_2$ fixing $p$, and a normalization of the affine generators $l_1^a$ and $l_2^a$ with 
$l_1^a(p)=l_2^a(p)$, such that, \textit{if $N_2$ lies at the future of $N_1$ near $p$}, then  (compare with (\ref{comp1}))
\begin{equation} \label{ineq1}
\Theta_{{N_2}_{l_2}} |_{\Sigma_2}(h(x)) \geq \Theta_{{N_1}_{l_1}} |_{\Sigma_1}(x)
\end{equation}
 and moreover 
\begin{equation}\label{ineq2}
\Theta_{{N_2}_{l_2}} |_{\Sigma_2}(h(x)) \leq \Theta_{{N_1}_{l_1}} |_{\Sigma_1}(x) \Rightarrow N_1=N_2 \;\; \text{ near } p.
\end{equation}
The strength of this result in general situations is not clear, because it assumes that the generators  have been scaled in a particular way that 
happens to depend on the  chosen $P$ hypersurface, a fact that  makes it difficult  to anticipate if the inequality (\ref{ineq2}) holds. 
An obvious exception is when the divergences have opposite signs, since 
in this case  the scaling is irrelevant. In this case we are led to \cite{Galloway1,Galloway2}:

\begin{thm} \nonumber    (G. Galloway \cite{Galloway1} \cite{Galloway2})
Let $N_1$ and $N_2$ be smooth null hypersurfaces in a spacetime M. 
Suppose,
\begin{enumerate}
\item  $N_1$ and $N_2$ are tangent  at a point  $p$ and $N_2$ lies to the future side of $N_1$ near $p$, and
\item  the divergences $\Theta_1$ of $N_1$  and $\Theta_2$ of $N_2$  satisfy $\Theta_2 \leq  0 \leq \Theta_1$.
\end{enumerate}
Then $N_1$ and $N_2$ coincide near $p$ and this common null hypersurface has  $\Theta=0$. 
\end{thm}

Consider now  the situation of 
a \textit{surface} $S$ intersecting \textit{tangentially} at  $p$ a null hypersurface $N_1$ with affine null field $l_1^a$. Since 
$T_pS \subset T_p N_1 = (\mathbb{R} l_1^a(p))^\perp$,  $l_1^a(p)$ points along  one of the two future null directions 
orthogonal to $T_pS$, then the  geodesic it generates belongs to both $N_1$ and $N_2:=N^S_{\ell}$: the null geodesic bundle 
of $S$ with any affine generators $l^a_2$ satisfying $l^a_2(p)=l_1^a(p)$ (here  $\ell^a := l_2^a \mid_S$). 
Under the hypothesis of Theorem \ref{main}, $S$ intersects tangentially 
the null hypersurface $N[g(p)]$, then so does its null bundle $N^S_\ell$, $\ell^a(p) = \nabla^a g(p)$. This leads  to the situation of Galloway's 
theorem with $N_2 \equiv N^S_\ell$ and $N_1 \equiv N[g(p)]$. The fact that $p$ is not \textit{any} critical point but a local maximum implies that 
$N_2$ lies to the future side of $N_1$ near $p$, meaning that  timelike curves from $N_1$ to $N_2$ in an open neighborhood of $p$  
must be future oriented. To prove this we need to show that $p$ is a local maximum of $g|_{N_2}$: to this end, take an open neighborhood $p \in U \subset S$ 
such that $g(q) \leq g(p)$ if $q \in U$. For $q \neq p$ either $g(q) < q(p)$, in which case there is a small open segment 
of the $N_2$ generator through $q$,  containing $q$,  where the condition $g<g(p)$ holds,  or  $g(q)=g(p)$, 
in which case $q$ is also a local maximum of $g|_S$, a segment around $q$ of the $N_2$ generator through $q$ is contained in $N_1$, and $g=g(p)$ 
along it. Thus, there is 
  an open neighborhood $\tilde U$ of $p$ in $N_2$  such that $g(r) \leq g(p)$ for $r \in \tilde U$, as we wanted to prove. 
  Note that,  since $N_1$ is the level set $g=g(p)$, 
   this implies that near $p$ $N_2$ is on the $g \geq g(p)$ side of $N_1$. To show that this is the future side, consider
  a timelike curve $[0,\tau_o] \ni \tau \to x^a(\tau)$ from $N_1$ to $\tilde U$. This must be future directed, as 
if it were past directed then $dg/d\tau= \nabla_a g \, d x^a / d\tau >0$ along the curve, which is inconsistent with $\tilde U 
 \ni g(x(\tau_o)) \leq g(x(0))=g(p)$.
Thus, the local maximum condition assures that conditions 1 in Galloway's theorem are fulfilled. Since $\Theta_1 = \Box g >0$, assuming 
that $S$ is trapped at $p$ would imply  $\Theta_1 <0$ near $p$ and  contradict Galloway's theorem. We then conclude that $S$ is not trapped at $p$. \\

\noindent
\underline{Example (continued)}: \textit{Figure \ref{xfig} shows the example 
we have been developing: the ellipse $S$ of equation (\ref{ellipse})  is the $t=0$ section of $N_2:=N^S_\ell$, equation (\ref{exNl}), which lies in front 
of  the level surface $N_1$ of $g=\sqrt{x^2+y^2}-t$ containing the point $p=(t=0,x=a,y=0)$. $N_1$ and $N_2$ are tangent 
along the generator through $p$ (shown in the figure). The timelike plane $P$, defined by $x=a$ (semi-transparent in the figure) 
is as required by Galloway's theorem, 
the sections $P \cap N^S_\ell =: \Sigma_2$ and $P \cap N[g(p)]=:\Sigma_1$ where the divergences  of these null surfaces are compared 
in  (\ref{ineq1}) are shown (here $h$ is the map connecting the sections along by integral lines of $\p_t$). Note that the inequality (\ref{ineq2}) holds 
for the entire sections, whereas the equality (\ref{se}) holds at $p$. }

\begin{figure}
\includegraphics[scale=0.4]{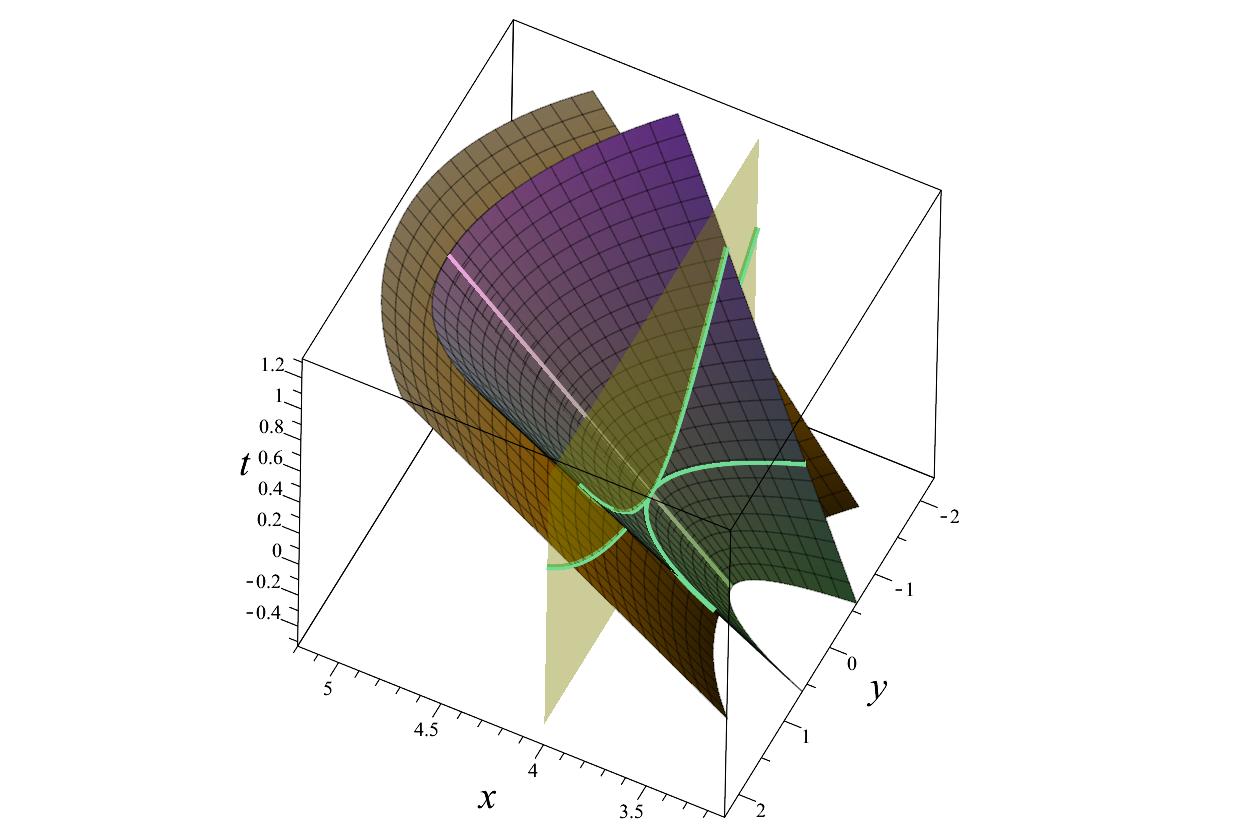}
\caption{Example developed through the text: (part of) the ellipse $S$ of equation (\ref{ellipse}) 
is shown as the $t=0$ subset of $N^S_\ell $, which appears in front of  the tangent null hypersurface 
$N[g(p)]$. The shared generator of these null hypersurfaces is shown, as well as 
the  timelike hyperplane $P$ (transparent)  used in Galloway's theorem and its intersection with the null hypersurfaces.
The point $p$ lies  in the intersection of all these surfaces and hypersurfaces.    In the figure, the ellipse semi-axes 
have values $a=4$ and $b=2$. \label{xfig}}
\end{figure}

\section{Applications} \label{apps}

The application of Corollary \ref{coro1} to a given spacetime $(M,g_{ab})$ is straightforward: 
given solution $g$ of the eikonal equation 
\begin{equation}\label{eiko}
g^{ab} \p_a g \, \p_b g =0,
\end{equation} 
either $g$ or $-g$ has  a future null gradient  in the region $D$ 
we are interested in. After solving this sign ambiguity (say, in favor of $g$) 
we find the subset $X_g \subset D$ defined by the condition that $g^{ab} \nabla_a \nabla_b g>0$. 
Since equation (\ref{eiko})  arises when solving the geodesic equation using the Hamilton-Jacobi method, 
it has been studied in detail for many background spacetimes.

\subsection{Minkowski spacetime}

The case where $M$ is the  $(n+1)-$dimensional Minkowski space $\mathbb{R}^{n+1}_1$, $n \geq 2$,
  offers the simplest application of Corollary \ref{coro1}. 
Assuming 
there is a closed trapped surface $S \subset M$ leads to a contradiction: take a global inertial frame with the $t$ axis away from $S$. 
In  spherical coordinates 
the metric is
\begin{equation}
ds^2=-dt^2 +dr^2 +r^2 d\Omega_{n-1}^2,
\end{equation}
$d\Omega_{n-1}^2$ the metric of the unit $(n-1)-$sphere. 
 The function  $g=r-t$ is $C^\infty$ in the domain $O$ defined by $r>0$ and has gradient 
$\p_t + \p_r$, which is future null. By hypothesis, $S \subset O$. However this is not possible 
since  $\nabla_a \nabla^a g = \nabla_a \nabla^a r = (n-1)/r>0$ in $O=X_g$. \\

Note: the most elegant proof that there are no closed trapped surfaces in an open stationary subset of a spacetime $(M,g_{ab})$  comes 
from choosing $\zeta^a$ in (\ref{evo}) a future timelike Killing vector field. If $S$ were trapped then $H^a$ would be future timelike 
and the integrand in (\ref{evo}) would be negative. However the area of $S$ is invariant under the flow of a Killing vector field, thus 
the left side of (\ref{evo}) is zero and we get a contradiction. 

\subsection{General spherically symmetric spacetimes with static regions } 
In advanced Eddington-Finkelstein 
coordinates, static regions of spherically symmetric spacetimes  are  given by the $f(r)>0$ sector/s of a  metric of the form 
\begin{equation}\label{sss}
ds^2 = -f(r) \, dv^2+ 2 k(r) \, dv dr + r^2  (d \theta^2 + \sin^2 \theta d\phi^2), \;\; r>0, \;\; -\infty < v < \infty, \;\; k(r)>0.
\end{equation}
We take the  time orientation  such that the globally defined null vector $O^a \p_a= -\p_r$ is future, then  the Killing vector $\p_v$ 
is future timelike  wherever  $f(r)>0$. We assume that  $f(r)$ has $n$ zeroes (not necessarily simple) 
at the positive values $r_1<r_2<...<r_n$ and that $f(r)>0$ for  $r>r_n$. 
We  make no assumptions on the  asymptotic behavior of $f(r)$ for large $r$. In particular,  (\ref{sss}) is not necessarily extendable 
to a black hole spacetime. 
We define the open sets $O_j$ by $r_j<r<r_{j+1}$, and add the special cases $O_0$ ($0<r<r_1$) and $O_n$ ($r>r_n$). \\
A calculation shows that the $v=v_o, r=r_o$ spheres have mean curvature vector field
\begin{equation}\label{smcvf}
H = \frac{2}{r_o k(r_o)} \p_v + \frac{2 f(r_o)}{r_o k(r_o)^2} \p_r,
\end{equation}
which is future timelike iff $f(r_o)<0$, so these spheres are trapped 
 iff $f(r_o)<0$,
 that is, in the non-static regions where the Killing vector field $\p_v$ is 
spacelike. 
We will use Corollaries \ref{coro1} and \ref{coro2} 
to show that no closed trapped surface can lie in a connected static $f>0$ region, and to analyze the possibility 
that these surfaces enter the static regions from the non-static ones. 

\subsubsection{Spherically symmetric $X$ sets}
The spherically symmetric solutions of the eikonal equation (\ref{eiko}) for the metric (\ref{sss}) are 
\begin{equation}\label{solus}
g_1(v,r)= F_1(v) \;\; \text{ and } \;\;  g_2(v,r)=F_2(v-2r^*(r)), \;\;\; \left( r^*(r) = \int^r \frac{k(r)}{f(r)} dr \right)
\end{equation}
$g_1$ is globally defined and $g_2$ has, in general, a domain with connected components 
the    $O_j$. 
The generators of the null level sets  of $g_1$ ($g_2$) are the incoming (outgoing) radial null geodesics.  \\

The function $g_1$ is of no use  because $(g^{ab} \p_b g_1)   \p_a = (F_1'(v)/k(r) )\, \p_r$ and
$g^{ab} \nabla_a \nabla_b g_1 = 2F_1'(v)/ (r k(r))$, 
so that $F_1'<0$ is required for $\nabla^a g_1$  to be future oriented, but then  $g^{ab} \nabla_a \nabla_b g_1$ is negative definite and 
 $X_{g_1}$ is empty. \\

For $g_2$, instead,  we find that
\begin{equation}
O^a \p_a g_2 = -\p_r g_2 = \frac{2k(r)}{f(r)} F_2'(v-2r^*(r)), \;\;\;\; g^{ab} \nabla_a \nabla_b g_2= -\frac{2}{r k(r)} F_2'(v-2r^*(r)).
\end{equation}
In a connected region where $f<0$ we can reason as above to prove that, whenever $\nabla^a g_2$ is future,  
 $X_{g_2}$  is empty. This is to be expected, since $f<0$ regions contain trapped spheres. In an open  $O_j$ with $f>0$, instead, 
 $F_2'(v-2r^*(r))<0$ is required for $\nabla^a g_2$ to be future, so we may use any $F_2: \mathbb{R} \to \mathbb{R}$ 
with $F_2'$ negative definite. As a result, $g^{ab} \nabla_a \nabla_b g_2$ will be positive in the entire $O_j$ region, proving that 
 no closed trapped surface is contained in $O_j$.  This rules out closed trapped surfaces 
 in, e.g.,  the inner region $0<r_i$ or the outer domain $r>r_e$ of a Reissner-Nordström black hole, 
  the interior of  extremal charged black holes, 
 the outer domain of a Schwarzschild spacetime, Schwarzschild's  naked singularity, and $f>0$ regions of regular black holes.

 \subsubsection{Trapped surface barriers}
 
 If $r=r_n$, the largest zero of  $f(r)$, is simple, the 
 null hypersurface $r=r_n$    works as an event horizon, as it prevents closed trapped surfaces from $r<r_n$ 
 to enter the outer region $O_n$ defined by $r>r_n$  (recall that we assumed that $f(r)>0$ for $r>r_n$).
  This happens  \textit{no matter what 
 the asymptotic behavior of $f(r)$ is.} 
To prove this, pick $r_o>r_n$ to fix the definition of $r^*(r)=\int_{r_o}^r k(r') dr'/f(r')$. Note that $r^*(r)$ has  a logarithmic singularity: for $r \gtrsim r_n$ we have   
  $r^*(r) \simeq (k(r_n)/f'(r_n)) \ln (r /r_n-1)$). Choose any $F_2$ in (\ref{solus}) with negative definite derivative 
  and a finite limit $F_2(z)$ as $z \to \infty$, e.g.,  $F_2(z)=-\arctan(z)$, then 
\begin{equation}\label{ehk}
g_2 = -\arctan\left(v- 2 \int_{r_o}^r  k(r') dr'/f(r') \right)
\end{equation}
has $\nabla^a g_2$ future null and $\nabla^a \nabla_a g_2>0$ in the domain $O_n$. 
The important characteristic of  (\ref{ehk}) is that,  at any point $p \in O_n$,  
$g_2(p) > -\tfrac{1}{2} \pi = \lim_{r \to {r_n}^+}g_2(v,r)$, so that any closed surface entering $O_n$ would be forced to attain 
a local maximum of  $g_2|_{S \cap O}$. As this contradicts Corollary \ref{coro1}, we conclude that no closed trapped surface extends 
beyond the null hypersurface $r=r_n$. \\
We insist that this conclusion holds no matter what the global structure of the spacetime is: the null hypersurface $r=r_n$ works effectively 
as a black hole event horizon, in the sense that it is a barrier that closed trapped surfaces cannot cross. 
The proof just given should be compared with 
the proof (and hypotheses) that closed  trapped surface cannot trespass the event horizon of a black hole spacetime (see, e.g., Proposition 12.2.2 
in \cite{waldbook}). \\

Would a similar  argument  prove that closed trapped surfaces cannot enter an $f>0$ region $O_{j-1}$ 
from an $f<0$ region $O_j$?
The answer is in the negative: trespassing to the left is not forbidden. 
Assume that $f$ has a simple zero at  $r_{j}$ and $f<0$ in $O_j$. 
The argument above implies that a closed trapped surface in $O_j$ cannot trespass $r_{j+1}$ and end within $O_{j+1}$,   
as it would be forced to have a local maximum of a suitable $g_2$ in $O_{j+1}$. However, 
there is no obstruction from Corollary \ref{coro1} for  such a surface to enter the 
region $O_{j-1}$ and end there. This is so because $r^*(r) \to - \infty$ as $r \to r_{j}^-$ and, 
since $F_2$ in (\ref{solus}) has negative derivative, an appropriate $g_2$ would now be forced to have a local \textit{minimum} in 
$O_{j-1}$, and this does not conflict Corollary \ref{coro1}.

\subsubsection{Non spherically symmetric $X$ sets}

The eikonal equation (\ref{eiko}) admits a three-parametric, separable solution on the background (\ref{sss}):
\begin{equation}\label{sgf}
g=E \left( v-\int \!\frac{k(r)}{f(r)}\,dr
 \right) +\Phi\,\phi+ s_1\,\int \!{\frac{k(r) \sqrt {{E}^{2}{r}^{2}- {\alpha}^{2}f
 \left( r \right) }}{f \left( r \right) r}}\,dr+s_2\,\int \frac {\sqrt {{\alpha}^{2}  \sin^2  \theta
 -{\Phi}^{2}}}{\sin  \theta }
\,d\theta
\end{equation}
where $s_1=\pm 1$ and $s_2=\pm 1$ are independent signs. Note that setting $\alpha=0=\Phi$ we get 
functions of the form (\ref{solus}). In this section  we  look for  obstructions for closed trapped 
surfaces in $f<0$ regions (where trapped \textit{sphere} exist, see (\ref{smcvf})),
 then the integrands in the $r$ integrals are well defined for any value of $(E,\Phi,\alpha)$ and 
$s_1=-1$ guarantees that $\nabla^a g$ is future oriented, so we make this choice. 
Besides, we require that $\Phi^2/\alpha^2<1$. The domain 
of $\theta$ in (\ref{sgf}) is restricted around the equator by
\begin{equation} \label{domti}
\sin^2 \theta > \frac{\Phi^2}{\alpha^2}.
\end{equation}
The condition $g^{ab} \nabla_a \nabla_b g >0$ gives an  $X$ set that is invariant under the flow of the Killing vector fields $\p_v$ and $\p_\phi$: 
\begin{equation} \label{xsss}
X_g=\left\{ (v,r,\theta,\phi) \; \big| \;  \frac{ 4\,{E}^{2}{r}^{2}- {\alpha}^{2}r  f'(r)-2\,{\alpha}^{2} f
 \left( r \right)}{\sqrt {{E}^{2}{r}^{2}-f \left( r \right) {\alpha}^{2}}} <2\,s_2\, k(r)\; \alpha^2 \frac{\cos \left( \theta \right) }{\sqrt {{\alpha}^{2} \sin^2 \theta -{
\Phi}^{2}}}\right\}
\end{equation}
Note that the inequality that defines $X_g$ is invariant under the simultaneous change $\theta \to \pi-\theta$, $s_2 \to -s_2$: 
if $(v,r,\theta,\phi)$ belongs to the $X_g$ set defined by using $s_2=1$ in $g$, then $(v,r,\pi-\theta,\phi)$ belongs to the $X_g$ set defined by $s_2=-1$.  
This ``mirror'' $X$ set is to be expected from the symmetries of the metric.  

\subsubsection{Schwarzschild black hole interior}

Taking $s_2=1$, $k(r)=1$ and $f(r)=1-2M/r$ in (\ref{xsss}) gives $X_g$ defined by 
\begin{equation} \label{xsch1}
{\frac { \left( M-r \right) {\alpha}^{2}+2\,{E}^{2}{r}^{3}}{{\alpha}^{2} \sqrt {{E}^{2}{r}^{3}+ \left( -r+2\,M \right) {\alpha}^{2}}\sqrt {r}}}
<\frac{\cos \theta}{\sqrt{\alpha^2 \sin^2 \theta-\Phi^2}}
\end{equation}
The left hand side above grows with $E^2$, so we get the largest  $X$ set by setting $E=0$, 
\begin{equation}
\frac{M-r}{\sqrt{r} \sqrt{\alpha^2 (2M-r)}} < \frac{\cos \theta}{\sqrt{\alpha^2 \sin^2 \theta -\Phi^2}}
\end{equation}
Note that the open set (\ref{xset}) found in \cite{wald}, which was the initial motivation for this work, 
 corresponds to the particular case $\alpha=1, \Phi=0$ above.

\subsection{Kerr spacetime}

Consider sub-extreme ($0<a<M$) Kerr spacetime in advanced coordinates
\begin{equation}\label{adv}
ds^2 =  -\left(1- \frac{2Mr}{\rho^2} \right) dv^2+\rho^2 d\theta^2 + \left[ r^2+a^2+\frac{2Mr a^2 \sin^2\theta}{\rho^2} \right] \sin^2\theta  d\varphi^2\\
-\frac{4Mar \sin^2 \theta}{\rho^2} dv \, d\varphi + 2 \, dv \,dr -2 a \sin^2 \theta \ dr \, d\varphi. 
\end{equation}
Here
\begin{equation}
 \rho^2 = r^2+a^2 \cos^2 \theta, \;\; -\infty <,v,r < \infty, 
 \end{equation}
 and $\theta,\phi$ are the standard coordinates of $S^2$. The time orientation is such that 
the null vector 
\begin{equation}
O=-\p_r
\end{equation}
 is future pointing. \\

If $r_o>0$, the closed surfaces $r=r_o, v=v_o$ are spheres with a non standard $(++)$ metric. A calculation of the mean curvature vector 
field show that these are trapped iff $r_-<r_o<r_+$, where $r_\pm = M\pm \sqrt{M^2-a^2}$ are the inner and outer horizons. \\

The eikonal equation admits separable solutions thanks to the  Killing vector fields $\p_v$ and $\p_\varphi$,  
(to which the constants $E$ and $\Phi$ below are related) and a Killing tensor (to which the constant $Q$ below is related). 
This can be  written as 
\begin{multline}\label{eikerr}
g=E \left( v-\int \!{\frac {{a}^{2}+{r}^{2}}{{a}^{2}-2\,Mr+{r}^{2}}}
\,{\rm d}r \right) +\Phi\, \left( \varphi-\int \!{\frac {a}{{a}^{2}-2\,Mr
+{r}^{2}}}\,{\rm d}r \right) \\
+s_1\,\int \!{\frac {\sqrt {{E}^{2}{
r}^{4}+ \left( {E}^{2}{a}^{2}-{\Phi}^{2}-Q \right) {r}^{2}+2\,M
 \left( {E}^{2}{a}^{2}-2\,E\Phi\,a+{\Phi}^{2}+Q \right) r-Q{a}^{2}}}{{
a}^{2}-2\,Mr+{r}^{2}}}\,{\rm d}r\\
+s_2\,\int \!\sqrt {Q+{E}^{2} a^2 
  \cos^2  \theta    -{\Phi}^{2}
  \cot^2  \theta   }\,{\rm d}\theta,
\end{multline}
where $s_j=\pm1 $ are independent signs. Since $g$ has three parameters, there are plenty of possibilities to 
explore. Any   choice of $(E,\Phi,Q)$ restricts the domain of (\ref{eikerr}) in a way that $v$ and $\phi$ are unconstrained and $r$ and 
$\theta$ are limited by the conditions that the arguments of the square roots in (\ref{eikerr}) be positive. This immediately tells us 
that $\nabla^a g$ will be future between horizons (respectively outside this region) if $s_1=-1$ ($s_1=1$). \\

A natural question to ask is if there is an obstruction for closed trapped surfaces between horizons that generalizes  (\ref{xset}) 
to the rotating case. 
To answer this question we set $E=0, \Phi=0, Q=1$ and (as explained above), $s_1=-1$ in (\ref{eikerr}). This gives 
\begin{equation}
g =s_2 \; \theta + \arctan \left( \frac{r-M}{\sqrt{2Mr+a^2-r^2}} \right).  
\end{equation}
For the sign $s_2$  comments analogous to those following equation (\ref{xsss}) apply. 
We will set $s_2=1$, then 
the excluded region, defined by the condition $g^{ab} \nabla_a \p_b g\geq 0$ is 
\begin{equation}
X_g = \left\{ (t,r,\theta,\phi) \; | \; r_- <r<r_+, \cot(\theta) >\frac{M-r}{ \sqrt{2Mr-a^2-r^2)}}  \right\}
\end{equation}
To analyze the effect of the rotation parameter 
$a$ note that for $r<M$ the condition on $\theta$ becomes more restrictive for larger $a$, whereas for 
$M<r<r_+$ becomes less restrictive. In any case, as $r \to r_+$ from the left, 
the whole range of $\theta$ is allowed, as happens in the Schwarzschild case.

\section{Acknowledgements} 
I thank Gregory Galloway for clarifications regarding his results in \cite{Galloway1,Galloway2} 
and Martin Dotti for his encouragement in difficult times. A  comment made by a  referee on a previous version of the 
manuscript hinted  to the current proof of the main results, which is shorter and self-contained.

\appendix

\section{Notation}

By a \textit{spacetime} $M$ we mean an $n+1$ dimensional Lorentzian  manifold, $n \geq 2$,
 which we assume oriented and time oriented. We use the mostly plus signature. A null vector is a zero norm nonzero vector. 
A \textit{null hypersurface} is an $n$ dimensional submanifold $N \subset M$ with degenerate induced metric. 
It can be shown that $N$ is a bundle of null geodesics,  called \textit{generators} of $N$. We use $l^a$ for a future  affine tangent 
to the generators of $N$. In Lemma \ref{divL}, we use a  vector field extension $\mathcal{L}^a$ 
 of $l^a$ away from $N$ which is null everywhere (in general, $\mathcal{L}^b \nabla_b \mathcal{L}^a=0$ only on points of $N$). 
 A \textit{surface} $S \subset M$ is a codimension 2 
submanifold with positive definite induced metric. $\ell^a$ and $\mathcal{k}^a$ are two future null vector fields on the normal bundle $(TS)^\perp$,  
normalized as $\mathcal{k}^a \ell_a=-1$; they are unique up to flipping and re-scaling $\ell^a \to \lambda \ell^a$, 
$\mathcal{k}^a \to \lambda^{-1} \mathcal{k}^a$, $\lambda$ a positive function on $S$. We call $N_{\ell}^S$ the bundle of null geodesics 
from $S$ with initial condition $\ell^a$ and $l^a$ the tangent to these geodesics (then $l^a|_S=\ell^a$). 
$N_{\ell}^S$ is a null hypersurface (near $S$). 
$N_{\mathcal{k}}^S$  and $k^a$ are defined analogously.   $\{e_i, i=1,2,...,n-1\}$ is a (possibly orthonormal) local basis of vector fields 
of $TS$.  In the proof of Theorem \ref{main}, $\{E_i, L^a, K^a, i=1,2,...,n-1\}$ 
is a pseudo-orthonormal basis of vector fields defined on an open neighborhood $O \subset M$ of $ p \in S$, 
 which agrees with $\{e_i, \ell^a, \mathcal{k}^a, i=1,2,...,n-1\}$ on $S$.

\end{document}